\documentclass[11pt]{article}

\usepackage[T1]{fontenc}
\usepackage{textcomp}
\usepackage{palatino}
\usepackage{mathpazo}
\usepackage{stmaryrd}
\usepackage{hyperref}
\hypersetup{pdfpagemode=UseNone}

\usepackage{amsfonts}
\usepackage{amssymb}
\usepackage{amsmath}
\usepackage{latexsym}
\usepackage{amsthm}
\usepackage{eepic}
\usepackage{sectsty}
\usepackage{graphicx}
\usepackage[usenames]{color}

\newtheorem{theorem}{Theorem}
\newtheorem{lemma}[theorem]{Lemma}
\newtheorem{prop}[theorem]{Proposition}
\newtheorem{cor}{Corollary}
\theoremstyle{definition}
\newtheorem{definition}{Definition}

\newtheorem{fact}{Fact}

\newtheorem{corollary}{Corollary}

\newcommand{\tinyspace}{\mspace{1mu}}

\newcommand{\snorm}[1]{\lVert\tinyspace#1\tinyspace\rVert}

\newcommand{\ceil}[1]{\left\lceil #1 \right\rceil}

\newcommand{\defeq}{\stackrel{\smash{\text{\tiny def}}}{=}}
\newcommand{\tr}{\operatorname{Tr}}

\newcommand{\ip}[2]{\left\langle #1 , #2\right\rangle}

\newcommand{\equil}[1]{\check{#1}}
\def\({\left(}
\def\){\right)}
\def\I{\mathbb{1}}

\newcommand{\fid}{\operatorname{F}}
\newcommand{\setft}[1]{\mathrm{#1}}
\newcommand{\lin}[1]{\setft{L}\left(#1\right)}
\newcommand{\density}[1]{\setft{D}\left(#1\right)}

\newcommand{\pos}[1]{\setft{Pos}\left(#1\right)}

\newenvironment{mylist}[1]{\begin{list}{}{
    \setlength{\leftmargin}{#1}
    \setlength{\rightmargin}{0mm}
    \setlength{\labelsep}{2mm}
    \setlength{\labelwidth}{8mm}
    \setlength{\itemsep}{0mm}}}
    {\end{list}}

\newcommand{\class}[1]{\textup{#1}}

\def\X{\mathcal{X}}
\def\Y{\mathcal{Y}}
\def\Z{\mathcal{Z}}

\def\A{\mathcal{A}}

\def\Q{\mathcal{Q}}

\begin{document}

%-----------------------------------------------------------------------------%
\title{\bf Equilibrium Value Method for the Proof of QIP=PSPACE}
%-----------------------------------------------------------------------------%

\author{%
  Xiaodi Wu
  \footnote{The work was completed when the author was visiting the Institute for Quantum Computing , University of Waterloo as a research assistant. } \\
  \it \small Institute for Quantum Computing,  University of
  Waterloo,  Ontario, Canada, \\
  \it \small Department of Electrical Engineering and Computer Science, University of Michigan, Ann Arbor, USA
}

%\date{Oct, 2009}

\maketitle

\begin{abstract}
We provide an alternative proof of \class{QIP}=\class{PSPACE} to the
recent breakthrough result ~\cite{JainJUW09}. Unlike solving some
semidefinite programs that captures the computational power of
quantum interactive proofs, our method starts with one
\class{QIP}-Complete problem which computes the diamond norm between
two admissible quantum channels. The key observation is that we can
convert the computation of the diamond norm into the computation of
some equilibrium value. The later problem, different from the former
semidefinite programs, is of better form, easier to solve and could
be interesting for its own sake. The multiplicative weight update
method is also applied to solve the equilibrium value problem,
however, in a relatively simpler way than the one in the original
proof ~\cite{JainJUW09}. As a direct byproduct, we also provide a NC
algorithm to compute the diamond norm of a class of quantum
channels. Furthermore, we provide a generalized form of equilibrium
value problems that can be solved in the same way as well as
comparisons to semidefinite programs.
\end{abstract}

%-----------------------------------------------------------------------------%
\section{Introduction} \label{sec:introduction}
%-----------------------------------------------------------------------------%
The \emph{interactive proof system} model, which extends the concept
of \emph{efficient proof verification}, has gradually become a
fundamental notion in the theory of computational complexity since
its introduction ~\cite{GoldwasserMR85, Babai85} in the mid 1980s.
In this model, a computationally bounded \emph{verifier} interacts
with a \emph{prover} with unbounded computational power in one or
more rounds. The prover wants to convince the \emph{verifier} to
\emph{accept}(\emph{reject}) the input, and the verifier will make
its decision based on the interacting process.

The expressive power of this kind of interactive proof system model
with at most polynomial rounds of communications is characterized
~\cite{Shen92,Shamir92,LundFKN92} by the well-known relationship
\[
    \text{IP=PSPACE}
\]
through the technique commonly known as \emph{arithmetization}. Many
variants of the interactive proof system model have been studied by
introducing new ingredients, such as the public-coin interactive
proofs~\cite{Babai85,BabaiM88,GoldwasserS89}, multi-prover
interactive proofs~\cite{Ben-OrGKW88}, zero-knowledge interactive
proofs~\cite{GoldwasserMR85,GoldreichMW91} as well as the
competing-prover interactive proofs ~\cite{FeigeK97}.

This paper mainly works with the \emph{quantum interactive proof
system}, which is defined~\cite{KitaevW00,Watrous99-qip-focs} in a
similar way to ordinary interactive proof systems except the
verifier and the prover have access to quantum computers. Similar to
the classical cases, several variants of quantum interactive proof
systems have been studied, including the ordinary quantum
interactive proofs~\cite{Watrous99-qip-focs,KitaevW00}, public-coin
quantum interactive proofs~\cite{MarriottW05}, zero-knowledge
quantum interactive
proofs~\cite{Watrous09a,Kobayashi08,HallgrenKSZ08}, multi-prover
quantum interactive proofs~\cite{KempeKMV09,KobayashiM03} and the
competing-prover quantum interactive proofs
~\cite{Gutoski05,GutoskiW05,GutoskiW07}. The complexity class
\class{QIP} known as the problems having quantum interactive proof
systems satisfies~\cite{KitaevW00}
\[
        \text{PSPACE=IP}\subseteq\text{QIP}\subseteq\text{EXP}
\]
Along with the introduction of the complexity class \class{QIP},
several complete problems for this complexity class have been
discovered. The first complete problem, called \emph{close images},
was first proposed in 2000~\cite{KitaevW00}. Several relevant
problems which can be reduced to \emph{close images} were later
discovered~\cite{RosgenW05,Rosgen08}. Especially, the \emph{quantum
circuits distinguishability} problem, which was
proved~\cite{RosgenW05} to be QIP-complete, serves as our start
point to prove \class{QIP}=\class{PSPACE}.

Recently, a big breakthrough~\cite{JainJUW09} that proves
\class{QIP}=\class{PSPACE} uses the primal-dual
approach~\cite{AroraK07} based on the \emph{multiplicative weights
update method} to solve a certain kind of semidefinite programs that
characterize the computational power of \class{QMAM}. The latter
complexity class was proved~\cite{MarriottW05} to have equivalent
expressive power as \class{QIP}. The multiplicative weights update
method is a well-known framework (or meta-algorithm) which
originates in many fields. Its matrix version, which was recently
developed and discussed in a survey paper~\cite{AroraHK05} and the
PhD thesis of Kale~\cite{Kale07}, was shown to be a great success in
extending the potential applications of this famous framework.
Particularly, a combinatorial primal-dual approach for solving
semidefinite programs (SDP) was proposed~\cite{AroraK07} based on
the matrix multiplicative weights update (MMW) method. Under mild
conditions, the primal-dual approach can be used to improve the time
performance of many known approximation algorithms via semidefinite
program relaxations. The main advantage of the MMW method together
with the primal-dual SDP solver for the purpose of simulating
quantum complexity classes is that the resultant algorithm can be
easily implemented efficiently in parallel \footnote{We call any
algorithm efficient in parallel if it is in \class{NC}. However, in
our case, the size of the matrix representation of any quantum
system is exponential in the input size. Thus, \class{NC(poly)} is
considered as the final complexity class.}. By making use of the
known result \class{NC(poly)}=\class{PSPACE}~\cite{Borodin77}, we
can solve these SDPs in \class{PSPACE} and hence show that some
quantum complexity class is contained by \class{PSPACE}. Before the
proof of \class{QIP}=\class{PSPACE}~\cite{JainJUW09}, similar ideas
were applied to show the containment \class{QIP(2)}$\subseteq$
\class{PSPACE}~\cite{JainUW09} and \class{QRG(1)}$\subseteq$
\class{PSPACE}~\cite{JainW09}.% were proved in this particular way.

%Before the proof of \class{QIP}=\class{PSPACE}, it was also
%applied~\cite{JainUW09,JainW09} to show the containment
%\class{QIP(2)}$\subseteq$ \class{PSPACE} and
%\class{QRG(1)}$\subseteq$ \class{PSPACE}. The advantage of this
%method is that it can be easily parallelized for some kinds of
%semidefinite programs. Making use of the known result
%\class{NC(poly)}=\class{PSPACE}~\cite{Borodin77}, we can solve these
%semidefinite programs in \class{PSPACE}.

Unlike proving the result based on the formulation of the definition
of the computational class QMAM, our proof starts with one
QIP-Complete problem. The problem called \emph{quantum circuits
distinguishability} \footnote{It's been observed that one can also
start with another QIP-Complete problem \emph{close
images}~\cite{KitaevW00} or the protocol to simulate \class{QIP}
with competing provers in~\cite{GutoskiW05} to prove the same result
in almost the same way. However, the particular choice here connects
our algorithm to the computation of diamond norm which is of
independent interest.} computes the diamond norm between two
mixed-state quantum circuits. If we were again to directly compute
the diamond norm by its definition, we would encounter some SDPs or
convex programs which are more complicated than the one we would
have if start with QMAM. Although time-efficient algorithms have
already been proposed
%it is known that some kinds of convex  programs
~\cite{Watrous09b,BATS09} to approximate the diamond norm, it is
unknown whether these methods can also be space-efficient, namely
running in \class{PSPACE}. The crucial observation here, to
circumvent the problem above, is to change the form of the diamond
norm before the computation.
%Our method, which is quite different from the
%convex program method that comes from the definition of the diamond
%norm, is first to change the diamond norm's form a bit.
The resultant problem (see \textbf{Theorem}~\ref{thm:conversion})
has a very neat form and can be expressed as an equilibrium value.
Most importantly, as we will see later, there is a space-efficient
algorithm to solve the latter problem in \class{PSAPCE}. To our
knowledge, this conversion for the first time establishes the
connection between the computation of the diamond norm and the
computation of some equilibrium value. Precisely,
Theorem~\ref{thm:conversion} claims that under certain conditions
the gap between two promises of the diamond norm can be transferred
to the gap between two promises of the constructed equilibrium
value.

%Besides being an important step for our main theorem, it is
%interesting and could be useful for problems related to diamond
%norms.

The equilibrium value, or better known as the value with
\emph{minimax-maximin} form, is an important concept in theoretical
computer science. For instance, many game theory related problems
can be characterized naturally in this form. The fact that we can
exchange the positions of \emph{min} and \emph{max} in any
equilibrium value makes the problem well structured and provides a
simpler iterative algorithm (also based on multiplicative weight
update method) to approximate the equilibrium value than the one for
SDPs. Similar ideas were known in the study of game
theory\footnote{See the survey~\cite{AroraHK05} for more reference.}
before, and was applied~\cite{JainW09} in the proof of
\class{QRG(1)}$\subseteq$\class{PSPACE}. Due to the merits of the
equilibrium value problem, the converted problem from the diamond
norm in our paper has a relatively simpler solution (see
\textbf{Theorem}~\ref{thm:algorithm}) than the SDP considered
previously~\cite{JainJUW09}. As a result, our main theorem provides
a simplified proof for the following fact.

%In order to solve the equilibrium value problem, we still need to
%apply multiplicative weights update method. However, since this
%problem is more structured, the multiplicative weights update method
%plays a simpler role in the whole algorithm and its purpose is more
%straightforward and easier to understand. We will get more
%understanding about this through the discussion of this paper.

\begin{cor}
  \class{QIP}=\class{PSPACE}
\end{cor}

As a sequence of the connection we build between the computation of
the diamond norm and the computation of some equilibrium value, we
also demonstrate how our algorithm can be used to approximate the
diamond norm to high precision of a class of channels efficiently in
parallel. Precisely, we show

\begin{cor}Given the classical description of any two
admissible quantum channels $\Q_0, \Q_1$, the diamond norm of their
difference $\snorm{\Q_0-\Q_1}_\diamond$ can be approximated in
\class{NC} with inverse poly-logarithm precision\footnote{Here the
input size is exactly the size of those matrices representing two
channels. Thus, the precision scales down to inverse
poly-logarithm.}.
\end{cor}

This result supplements the time-efficient algorithm for calculating
diamond norm in~\cite{BATS09,Watrous09b}. Although our algorithm
only works for a special class of channels (also one of the most
interesting cases), extensions of the current algorithm for a larger
class of channels could be obtained if more complicated analysis is
involved.

It is interesting to compare the proof of \class{QIP}=\class{PSPACE}
in this paper and the one in~\cite{JainJUW09}. Our comparison
represented here is threefold. First, obviously the two approaches
diverge at the start point. However, this difference is actually
subtle. As we mentioned before, the QIP-Complete we considered could
be replaced by the \emph{close images} problem almost with no change
of the latter proof. If one investigates the result
\class{QIP}=\class{QMAM} carefully, one will find this equivalence
also comes from the \emph{close images} problem. Recall that the
diamond norm problem is more naturally formulated as some SDPs or
convex programs just as QMAM does. It thus seems like we
\emph{deliberately} formulate the problem by an equilibrium value
instead of a more natural formulation. We consider this as the main
difference between the two approaches.

%Although the two proofs both exploit the multiplicative weight
%update method, the reasons why the proofs work are quite different
%in several aspects. First, the start points of both proofs are
%different. Our proof starts from one \class{QIP}-Complete problem
%while the other one starts with the protocol of \class{QMAM}.
%Second, the resultant formulations are different. Our formulation is
%an equilibrium value problem while the other one becomes a
%semidefinite programming problem. Those two different points imply
%that no matter what kind of tools we might exploit later, the
%abstract mathematical problems after formulation are already quite
%different.

Second, different formulations hence lead to the need of algorithms
for different problems, SDPs and equilibrium value problems in our
case. Due to the relation \class{PSPACE}=\class{NC(poly)}, it
suffices to find algorithms that are efficient in parallel.
Fortunately, such algorithms for both problems can be obtained based
on the matrix multiplicative weight update method. Nevertheless, the
two algorithms are quite different in several aspects. We refer
curious readers to Kale's thesis~\cite{Kale07} for complete details,
while a brief comparison can be found below. We will refer the
algorithm for SDP in Kale's thesis as the primal-dual SDP solver
because there are indeed other methods for solving SDP also based on
the matrix multiplicative weight update method. Conceptually, the
primal-dual SDP solver exploits the duality between the primal and
dual problems of a certain SDP while minimax-maximin equality is
made use of for equilibrium value problems. It turns out the minimax
relation gives a simpler proof of the correctness of the algorithm
than the duality relation. Technically, both algorithms require
efficient implementation of some oracles. For those SDPs and
equilibrium values about QIP, the oracle for the equilibrium value
problem is easier to design than the one for SDP.\footnote{ It is
not easy to compare directly since the algorithm in~\cite{JainJUW09}
unpacks everything and only uses one sub-routine, namely
\emph{projection onto positive subspace} (the same as the one in our
algorithm). Nevertheless, if one rewrites the algorithm component by
component, one could find out that oracle is slightly harder to
solve. Moreover, additional assumptions like the invertibility of
some matrices are also necessary to solve that oracle. }
Furthermore, SDP solver faces an additional difficulty which is not
applicable to the equilibrium value problem. As one restriction of
the matrix multiplicative weight update method, any solution
obtained for SDP problems only satisfies the constraints
approximately. Namely, one needs to convert the raw solution into
exactly satisfiable solution. However, there is no control in
general about the change of the object function value after this
conversion. Therefore, converting approximate solutions to exact
satisfiable solutions without changing the object function value a
lot is another difficulty in designing SDP solver.

%Furthermore, we want to emphasize that it is the case that the ways
%how the multiplicative weight update method solves equilibrium value
%problems or semidefinite programming problems according to Kale's
%thesis~\cite{Kale07} are in fact different as well. We refer curious
%readers to Kale's thesis~\cite{Kale07} for details. Nevertheless, a
%brief comparison can be found below. Intuitively, both solutions
%include the exploitation of certain kind structure of the problem
%itself and the design of an efficiently implementable oracle with
%certain requirement. For semidefinite programming problems, the
%duality between the primal and dual problems is exploited while the
%minimax-maximin equality is made use of for equilibrium value
%problems. The correctness of the multiplicative weight update
%solution by using the latter structure is easier to prove. In
%addition, the oracle design for the concrete form of the equilibrium
%value in this paper is much simpler in the sense that only spectral
%decomposition is involved. These facts contribute to another reason
%why our proof is relatively simpler than the previous one.

Finally, it is hard to compare the performance (e.g, in terms of
time, space or other resources) of those methods for general
equilibrium value problems and SDPs. We do not even know to what
extent those methods can be applied to general equilibrium value
problems or SDPs. The analysis might heavily depend on the
particular form of the problem itself. However, some progress has
been made recently~\cite{Wu10,GutoskiW10} in finding efficient
algorithms for a larger class of equilibrium value problems and
SDPs. Particulary, there exists an equilibrium-value-based SDP
solver~\cite{Wu10} in addition to the primal-dual SDP solver. The
new SDP solver provides a generic way to design efficient oracles,
whereas a generic way of converting approximate solution to exactly
satisfiable solution remains unknown.

%%\textbf{TODO} It is an interesting problem to consider the general
%%similarities and differences between the semidefinite programs and
%%the equilibrium value problems. First, we need to notice that some
%%problems are naturally formed in only one form above. However, for
%%some other problems, we might have the freedom to formulate it in
%%either form we want. Since generally speaking the equilibrium value
%%problems are more structured, it might be a good idea to formulate
%%in that form if possible. Second, any semidefinite program is in
%%principle possible to be converted into a \emph{minimax} form by
%%using the duality relationship between the primal and dual problem
%%(although maybe not the specific \emph{minimax} form in this paper).
%%However, it is not necessarily true that the resultant
%%\emph{minimax} form is an equilibrium value since we need stronger
%%conditions to have an equilibrium value. Finally, we have to admit
%%that it is still a big open problem to parallelize the solution
%%efficiently to either semidefinite programs or equilibrium value
%%problems in general. Thus, the performance of each method might
%%heavily depend on the problem under consideration.

The rest of this paper is organized as follows. We briefly survey
some preliminaries which will be useful in our proof in
Section~\ref{sec:preliminaries}. The conversion from the
QIP-Complete problem to some equilibrium value problem lies in
Section~\ref{sec:conversion}, which is followed by the main proof of
\class{QIP}=\class{PSPACE} in Section~\ref{sec:MMW_PSPACE}. The
algorithm for computing the diamond norm is discussed in
Section~\ref{sec:diamond}.  We conclude the whole paper with the
summary, Section~\ref{sec:summary}, where we provide further
discussions about the equilibrium value problem and some open
problems. Before the readers move on to the next section, there is
one point to make clear. We will not take care of the precision
issues with the \class{NC} implementation in the main part of this
paper. Instead, we will assume such implementation can be made
exactly and deal with precision issues in
Appendix~\ref{sec:precision_issue}

%-----------------------------------------------------------------------------%
\section{Preliminaries} \label{sec:preliminaries}
%-----------------------------------------------------------------------------%
This section contains a summary of the fundamental notations about
the useful linear-algebra facts in quantum information.  For the
most part of this section, it is meant to make clear the notations
and the terminology used in this paper. For those readers who are
not familiar with these concepts, we recommend them to refer to
~\cite{Bhatia97,KitaevW02,NielsenC00,Watrous08}.

A \emph{quantum register} refers to a collection of qubits, usually
represented by a complex Euclidean spaces of the form
$\X=\mathbb{C}^\Sigma$ where $\Sigma$ refers to some finite
non-empty set of the possible states.

For any two complex Euclidean spaces $\X,\Y$, let $\lin{\X,\Y}$
denote the space of all linear mappings (or operators) from $\X$ to
$\Y$ ($\lin{\X}$ short for $\lin{\X,\X}$). An operator $A \in
\lin{\X,\Y}$ is a \emph{linear isometry} if $A^{\ast}A=\I_\X$ where
$A^{\ast}$ denotes the adjoint (or conjugate transpose) of $A$.
%
%The tensor product $\X \otimes \Y$ of vector space
%$\X=\mathbb{C}^\Sigma$ and $\Y=\mathbb{C}^\Gamma$ is associated with
%the space $\mathbb{C}^{\Sigma \times \Gamma}$. The tensor product of
%operators $A \in \lin{\X} $ and $B \in \lin{\Y}$ is defined to be
%the unique linear mapping that satisfies $(A\otimes B)(x \otimes
%y)=(Ax) \otimes (By)$ for all $x \in \X$ and $y \in \Y$.

 An operator $A\in \lin{\X}$ is \emph{Hermitian} if
$A=A^{\ast}$. The eigenvalues of a Hermitian operator are always
real. For $n=\dim{\X}$, we write
\[
  \lambda_1(A)\geq \lambda_2(A)\geq \cdots \geq \lambda_n(A)
\]
to denote the eigenvalues of $A$ sorted from largest to smallest. An
operator $P\in \lin{\X}$ is \emph{positive semidefinite}, the set of
which is denoted by $\pos{\X}$, if $P$ is Hermitian and all of its
eigenvalues are nonnegative, namely $\lambda_n(P)\geq 0$. An
operator $\rho \in \pos{\X}$ is a \emph{density operator}, the set
of which is denoted by $\density{\X}$, if it has trace equal to 1.
It should be noticed that a \emph{quantum state} of a quantum
register $\X$ is represented by a density operator $\rho \in
\density{\X}$.

%An operator $\Pi \in \pos{\X}$ is a \emph{projection} if $\Pi$
%projects onto some subspace of $\X$. Furthermore, such operators
%only have eigenvalues of 0 or 1.

The Hilbert-Schmidt inner product on $\lin{\X}$ is defined by
\[
   \ip{A}{B}=\tr{A^{\ast}B}
\]
for all $A,B \in \lin{\X}$.

%For spaces $\X$ and $\Y$, one can define the \emph{partial trace}
%$\tr_{\Y}: \lin{\X \otimes \Y} \rightarrow \lin{\X}$ to be the
%unique linear mapping that satisfies $\tr_{\Y}(A \otimes B)=(\tr A)
%B$ for all $A \in \lin{\X}$ and $B \in \lin{\Y}$.
%
% We refer to \emph{measurements}, or precisely POVM-type
%measurements as a collection of positive semidefinite operators
%\[
%   \{P_a : a \in \Sigma\} \subset \pos{\X}
%\]
%satisfying the constraint $\sum_{a \in \Sigma} P_a =\I_\X$. Here
%$\Sigma$ refers to a finite, nonempty set of \emph{measurement
%outcomes}. If a quantum state represented by $\rho \in \density{\X}$
%is measured with respect to this measurement, then each outcome $a
%\in \Sigma$ will be observed with probability $\ip{P_a}{\rho}$.

A \emph{super-operator} (or quantum channel) is a linear mapping of
the form
\[
  \Psi : \lin{\X} \rightarrow \lin{\Y}
\]
A super-operator $\Psi$ is said to be \emph{positive} if $\Psi(X)
\in \pos{\Y}$ for any choice of $X \in \pos{\X}$, and is
\emph{completely positive} if $\Psi \otimes \I_{\lin{\Z}}$ is
positive for any choice of a complex vector space $\Z$.  The
super-operator $\Psi$ is said to be \emph{trace-preserving} if
$\tr{\Psi(X)} =\tr{X}$ for all $X \in \lin{\X}$. A super-operator
$\Psi$ is \emph{admissible} if it is completely positive and
trace-preserving. Admissible super-operators represent the
discrete-time changes in quantum systems that, in principle, can be
physically realized.

One can also define the adjoint super-operator of $\Psi$, denoted by
\[
 \Psi^{\ast} : \lin{\Y} \rightarrow \lin{\X}
\]
to be the unique linear mapping that satisfies,
\[
 \ip{B}{\Psi(A)}=\ip{\Psi^{\ast}(B)}{A}
\]
for all operators $A\in \lin{\X}$ and $B \in \lin{\Y}$.

%The Choi-Jamiolkowski representation [cite] of super-operators is as
%follows.  Let $\Psi: \lin{\X} \rightarrow \lin{\Y}$ be a given
%super-operator, then the \emph{Choi-Jamiolkowski representation} of
%$\Psi$ is denoted
%\[
%  J(\Psi)= \sum_{i,j \in \Sigma} \Psi(\ket{i}\bra{j})\otimes
%  \ket{i}\bra{j} \in  \lin{\Y \otimes \X}
%\]
%For an admissible super-operator $\Psi$, the $J(\Psi)$ is positive
%semidefinite and $\tr_{\Y} J(\Psi)=\I_{\X}$.

The \emph{Stinespring representations} of super-operators is as
follows. For any super-operator $\Psi$, there is some auxiliary
space $\Z$ and $A,B \in \lin{\X, \Y \otimes \Z}$ such that
\[
   \Psi(X)=\tr_{\Z} A X B^{\ast}
\]
for all $X \in \lin{\X}$.  When $\Psi$ is admissiable, we have $A=B$
and $A$ is a linear isometry.

% The \emph{spectral norm} of an operator $A \in \lin{\X}$ is
%defined as
%\[
%   \snorm{A}=\max \{ \snorm{Ax} : x \in \X, \norm{x}\leq 1\}
%\]
%It is easy to see that $\snorm{A}$ is actually the maximum of the
%absolute value of all eigenvalues. For any $ A \in \pos{\X}$, we
%have $\snorm{A}=\lambda_1(A)$.

 A \emph{quantum circuit} is an acyclic network of \emph{quantum
gates} connected by wires. The quantum gates represent feasible
quantum operations, involving constant numbers of qubits. In a
\emph{mixed state quantum circuit}~\cite{AKN98}, instead of using
unitary operations as quantum gates, we allow the gates to be from
any set of quantum admissible operations. In this more flexible
circuit model, some part of the qubits might be discarded (or
\emph{traced out}) during the evolution of the circuit.

The \emph{trace norm} of an operator $A \in \lin{\X}$ is denoted by
$\snorm{A}_1$ and defined to be
\[
  \snorm{A}_1= \tr \sqrt{ A^\ast A}
\]
When $A$ is Hermitian, we have
\[
  \snorm{A}_1= \max\{ \ip{P_0-P_1}{A}: P_0,P_1 \in \pos{\X},
  P_0+P_1=\I_\X\}
\]

The diamond norm of a super-operator $\Psi: \lin{\X} \rightarrow
\lin{\Y} $ is defined to be
\[
   \snorm{\Psi}_{\diamond}=\max_{\snorm{X}_1 \leq 1} \snorm{\Psi
   \otimes \I_{\X}(X)}_1
\]
\noindent Because of including the effect of using entanglement
between the input and some auxiliary space, the diamond norm serves
as a good measure of the distinguishability between quantum
operations. Furthermore, we can show
\begin{fact} \label{fact:fmax} ~\cite{KitaevW02}
If a quantum channel $\Psi$ can be represented by $\Psi(X)=\tr_{\Z}
A X B^{\ast}$ where $A,B \in \lin{\X, \Y \otimes \Z}$, define the
new channels
\[
   \Psi_{A}(X)=\tr_{\Y} A X A^{\ast} \quad , \quad
   \Psi_{B}(X)=\tr_{\Y} B X B^{\ast}
\]
then the diamond norm of this channel $\Psi$ is equal to
\begin{equation} \label{eqn:diamond_norm}
\snorm{\Psi}_{\diamond}=\fid_{\max}(\Psi_A, \Psi_B)
\end{equation}
where
\[
 \fid_{\max}(\Psi_{A}, \Psi_{B})=\max\{ \fid(\Psi_{A}(\varrho), \Psi_{B}(\zeta)):\varrho,\zeta\in\density{\X} \}
\]
and
\[
  \fid(P,Q)=\snorm{\sqrt{P}\sqrt{Q}}_1
\]
which is a generalization of the fidelity between quantum states.
\end{fact}

%-----------------------------------------------------------------------------%
\section{Conversion of the QCD Problems to Equilibrium Value Problems}
\label{sec:conversion}
%-----------------------------------------------------------------------------%
%\emph{Quantum interactive proof} are interactive proof systems in
%which the prover and the verifier may exchange and process quantum
%information. [cite]. The complexity class \class{QIP} is defined as
%the class of the promise problems having quantum interactive proof
%systems. This complexity class captures the efficient proof
%verification process by a single prover with access to quantum
%resources.
%
%There are a few \class{QIP}-Complete problems discovered during the
%study of quantum interactive proof system. One \class{QIP}-Complete
%problem known as Quantum Circuit Distinguishability(QCD) problem
%[cite] will be the start point for our later proof of
%\class{QIP}=\class{PSPACE}. Precisely,

\begin{definition}[Quantum Circuit Distinguishability]
For any constant $a,b$, such that $0\leq b<a\leq 2$. We define a
promise problem $\class{QCD}_{a,b}$ as follows. Given the
description of any two mixed-state quantum circuits $\mathcal{Q}_0$
and $\mathcal{Q}_1$, which are admissible quantum channels from
$\lin{\X}$ to $\lin{\Y}$, exactly one of the following conditions
will hold:
\begin{enumerate}
  \item $\snorm{\mathcal{Q}_0-\mathcal{Q}_1}_{\diamond} \geq a$
  \item $\snorm{\mathcal{Q}_0-\mathcal{Q}_1}_{\diamond} \leq b$
\end{enumerate}
$\class{QCD}_{a,b}$ will \emph{accept} on the condition (1) and
\emph{reject} otherwise.
\end{definition}

It was proved by Rosgen \emph{et al.}~\cite{RosgenW05} that for any
constant $0<\varepsilon<1$, $\class{QCD}_{2-\varepsilon,
\varepsilon}$ is QIP-Complete. A careful reformulation of this
problem will enable us to rephrase this promise problem in term of
an equilibrium value problem.

Assume there exists some space $\Z \otimes \Q$ that will be
constructed later, we define a min-max value $\equil{\lambda}(\Xi)$
to be $\min_{\rho \in \density{\X_0 \otimes \X_1}} \max_{\Pi \in
\Gamma} \ip{\Pi}{\Xi(\rho)}$ where $\Xi$ is a linear super operator
mapping from $\lin{\X_0 \otimes \X_1}$ to $\lin{\Z \otimes \Q}$ and
$\Gamma=\{\Pi: 0\leq \Pi \leq \I_{\Z\otimes \Q}\}$. The $\X_0, \X_1$
in the above definition are isomorphic copies of $\X$.  Further
investigation shows the value $\equil{\lambda}(\Xi)$ is also an
\emph{equilibrium value}.

Given $\density{\X_0 \otimes \X_1}$ and $\Gamma$ are convex and
compact sets and the $\ip{\Pi}{\Xi(\rho)}$ is a bilinear function
over them, it follows from the well-known extensions of von'
Neumann's Min-Max Theorem ~\cite{vN28,Fan53} that
\begin{equation}\label{eqn:minmax_thm}
\equil{\lambda}(\Xi)=\min_{\rho \in \density{\X_0 \otimes \X_1}}
\max_{\Pi \in \Gamma} \ip{\Pi}{\Xi(\rho)}= \max_{\Pi \in \Gamma}
\min_{\rho \in \density{\X_0 \otimes \X_1}} \ip{\Pi}{\Xi(\rho)}
\end{equation}

The equilibrium value $\equil{\lambda}(\Xi)$ is the quantity
represented by the two sides of the above equation. Furthermore, any
pair $(\equil{\rho},\equil{\Pi})$ which makes the function reach the
equilibrium value is called the \emph{equilibrium point}; or,
equivalently, that
\[
\min_{\rho \in \density{\X}} \ip{\equil{\Pi}}{\Xi(\rho)} =
\ip{\equil{\Pi}}{\Xi(\equil{\rho})}=\max_{\Pi \in \Gamma}
\ip{\Pi}{\Xi(\equil{\rho})}
\]
The existence of the equilibrium point follows easily from Equation
[\ref{eqn:minmax_thm}]. Careful readers might notice the equilibrium
value's form defined in this paper is related but slightly different
from the one defined in the proof of
\class{QRG(1)}$\subset$\class{PSPACE}~\cite{JainW09}. In the
latter's definition $\Gamma$ is the set of the density operators.
Thus, the equilibrium value will be the largest eigenvalue or
$\mathcal{L}_\infty$ norm in some sense. However, our definition of
$\Gamma$ makes the equilibrium value be the summation of all
positive eigenvalues. Moreover, in the situation of later discussion
in this paper, the equilibrium value turns to be half the
$\mathcal{L}_1$ norm.

Our main theorem of this part says the two promises of any QCD
problem, or equivalently of any diamond norm of the difference of
two admissible channels, can be transferred to the two promises of
the value of $\equil{\lambda}(\Xi)$ where $\Xi$ can be constructed
efficiently from the input to that QCD problem. Precisely,

\begin{theorem}\label{thm:conversion}
For any instance of the $\class{QCD}_{a,b}$ problem, there exist
some space $\Z \otimes \Q$ and a linear super operator $\Xi$ from
$\lin{\X_0 \otimes \X_1}$ to $\lin{\Z \otimes \Q}$ where the space
$\X_0, \X_1$ are isomorphic copies of the space $\X$ such that
\[
\left\{
  \begin{array}{ll}
  \equil{\lambda}(\Xi)\leq \frac{\sqrt{4-a^2}}{2}, & \hbox{$\snorm{\mathcal{Q}_0-\mathcal{Q}_1}_{\diamond} \geq a$;} \\
  \equil{\lambda}(\Xi)\geq \frac{2-b}{2} , & \hbox{$\snorm{\mathcal{Q}_0-\mathcal{Q}_1}_{\diamond} \leq b$.}
  \end{array}
\right.
\]
where $\equil{\lambda}(\Xi)$ is the equilibrium value defined above.
Further more, such a super operator $\Xi$ can be constructed
efficiently in parallel from the input to the $\class{QCD}_{a,b}$
problem.
\end{theorem}

Before we get into the proof of the theorem, it might be helpful to
see where this theorem leads us to. Since for any $0 <\varepsilon
<1$ the $\class{QCD}_{2-\varepsilon,\varepsilon}$ is QIP-Complete,
we can choose a constant $\varepsilon'$ such that
$\equil{\lambda}(\Xi)$ is either at least $\frac{2-\varepsilon'}{2}$
or at most $\frac{\sqrt{4\varepsilon'-\varepsilon'^2}}{2}$ where
$\frac{2-\varepsilon'}{2} \geq
\frac{\sqrt{4\varepsilon'-\varepsilon'^2}}{2}$ and there is a
constant gap between the two promises. For example, if we choose
$\varepsilon=0.1$, then the two promises become
\[
  \text{either } \equil{\lambda}(\Xi) \geq 0.95 \text{ (namely, }  \snorm{\Phi}_\diamond\leq 0.1) \text{ or }
  \equil{\lambda}(\Xi) \leq 0.32 \text{ (namely, } \snorm{\Phi}_\diamond\geq 1.9)
\]
Thus, in order to simulate QIP, it suffices to compute
$\equil{\lambda}(\Xi)$ approximately to distinguish between the two
promises. This accomplishes the conversion we need for the next step
of the whole proof. A simple proof for
\emph{Theorem~\ref{thm:conversion}} is available below.

\begin{proof}
For any instance of $\class{QCD}_{a,b}$, we are given the classical
descriptions of two mixed-state quantum circuits $\mathcal{Q}_0$ and
$\mathcal{Q}_1$, which are admissable quantum channels from
$\lin{\X}$ to $\lin{\Y}$. Thus, we could describe the two circuits
using the Stinespring representation of quantum channels.
Precisely,\[
  \mathcal{Q}_0(X) = \tr_{\Z} (A_0 X
  A_0^{\ast}), \quad
  \mathcal{Q}_1(X) = \tr_{\Z} (A_1 X A_1^{\ast})
\]
where $\Z$ is the auxiliary space and $A_0,A_1 \in \lin{\X, \Y
\otimes \Z}$ are linear isometries.  Then, we have,
\[
 \Phi(X) \defeq{} \mathcal{Q}_0 (X)-\mathcal{Q}_1 (X) =\tr_{\Z} (A_0 X A_0^{\ast}-A_1
 X A_1^{\ast}) =\tr_{\Z \otimes \Q} (2C_0 X C_1^{\ast})
\]
where $\Q$ is a complex Euclidean space of dimension $2$ and $C_0,
C_1 \in \lin{\X, \Y \otimes \Z \otimes \Q}$. Moreover,
\[
 C_0=\frac{1}{\sqrt{2}}\left(
       \begin{array}{c}
         A_0 \\
         A_1 \\
       \end{array}
     \right),\quad \quad
 C_1^{\ast} =\frac{1}{\sqrt{2}}\left(
        \begin{array}{cc}
          A_0 & -A_1 \\
        \end{array}
      \right)
\]
It is easy to see that $C_0^{\ast}C_0=C_1^{\ast}C_1=\I_{\X}$ given
that $A_0^{\ast}A_0=A_1^{\ast}A_1=\I_{\X}$. To compute the diamond
norm of $\Phi$, we define
\[
  \Phi_{A}(X)  = \tr_{\Y} (C_0 X C_0^{\ast}), \quad
  \Phi_{B}(X)  = \tr_{\Y} (C_1 X C_1^{\ast}) \\
\]
Due to Fact~\ref{fact:fmax} , we have
\begin{equation}\label{eqn:diamond_one}
\snorm{\Phi}_{\diamond}=\fid_{\max}(2\Phi_{A}, 2\Phi_{B})
\end{equation}
 It is interesting and useful to see that we can use one density operator
$\rho \in \density{\X_0 \otimes \X_1}$ to represent $\varrho, \zeta
\in \density{\X}$ where the space $\X_0, \X_1$ are isomorphic copies
of the space $\X$.

Precisely, let
\[
\widetilde{\Phi_{A}}(X)=\Phi_A(\tr_{\X_1}(X)), \quad
\widetilde{\Phi_{B}}(X)=\Phi_B(\tr_{\X_0}(X))
\]

It is obvious that $\widetilde{\Phi_{A}},\widetilde{\Phi_{B}}$ are
admissible quantum channels from $\lin{\X_0 \otimes \X_1}$ to
$\lin{\Z \otimes \Q}$. Define
\begin{equation}\label{eqn:diamond_two}
 {\widetilde{\fid_{\max}}(2\widetilde{\Phi_{A}}, 2\widetilde{\Phi_{B})}}=\max\{ \fid(2\widetilde{\Phi_{A}(\rho)}, 2\widetilde{\Phi_{B}(\rho)}):\rho \in\density{\X_0 \otimes \X_1} \}
\end{equation}

By taking $\rho=\varrho \otimes \zeta$ in the Equation
[\ref{eqn:diamond_two}], we have
${\widetilde{\fid_{\max}}(2\widetilde{\Phi_{A}},
2\widetilde{\Phi_{B})}} \geq \fid_{\max}(2\Phi_{A}, 2\Phi_{B})$ . To
see the reverse side, we can take $\varrho=\tr_{\X_1}\rho$ and
$\zeta=\tr_{\X_0} \rho$ in Equation [\ref{eqn:diamond_one}]. Thus,
we have ${\widetilde{\fid_{\max}}(2\widetilde{\Phi_{A}},
2\widetilde{\Phi_{B})}}=\fid_{\max}(2\Phi_{A}, 2\Phi_{B})$. Namely,
\[
  \snorm{\Phi}_{\diamond}= \max\{
  \fid (2\widetilde{\Phi_{A}(\rho)},
2\widetilde{\Phi_{B}(\rho)}):
  \rho \in \density{\X_0 \otimes \X_1} \}
\]

Since in QCD problem we have the promise that either
$\snorm{\Phi}_{\diamond}\geq a$ or $\snorm{\Phi}_{\diamond}\leq b$.
Due to the Fuchs-van de Graaf Inequalities, for any $\varrho, \zeta
\in \density{\X}$ ,
\begin{equation} \label{eqn:fuchs_inequality}
 1-\frac{1}{2}\snorm{\varrho-\zeta}_1 \leq \fid(\varrho,\zeta) \leq
 \sqrt{1-\frac{1}{4}\snorm{\varrho-\zeta}_1^2}
\end{equation}
and let $\Gamma=\{ \Pi: 0\leq \Pi \leq \I_{\Z \otimes \Q} \}$. By
substituting $\fid (\widetilde{\Phi_{A}}, \widetilde{\Phi_{B}})$
into Eq [\ref{eqn:fuchs_inequality}] and make use of the fact $\fid
(2\widetilde{\Phi_{A}}, 2\widetilde{\Phi_{B}})=2\fid
(\widetilde{\Phi_{A}}, \widetilde{\Phi_{B}})$, then we have when
$\snorm{\Phi}_{\diamond}\geq a$,
\[
  \min_{\rho \in \density{\X_0 \otimes \X_1}} \max_{\Pi \in
  \Gamma}
  \ip{\Pi}{\widetilde{\Phi_{A}}(\rho)-\widetilde{\Phi_{B}}(\rho)} =
\min_{\rho \in \density{\X_0 \otimes \X_1}}  \frac{1}{2}
  \snorm{\widetilde{\Phi_{A}}(\rho)-\widetilde{\Phi_{B}}(\rho)}_1 \leq
   \frac{\sqrt{4-a^2}}{2}
\]
and when $\snorm{\Phi}_{\diamond}\leq b$,
\[
 \min_{\rho \in \density{\X_0 \otimes \X_1}} \max_{\Pi \in
  \Gamma}
  \ip{\Pi}{\widetilde{\Phi_{A}}(\rho)-\widetilde{\Phi_{B}}(\rho)} =
\min_{\rho \in \density{\X_0 \otimes \X_1}}  \frac{1}{2}
  \snorm{\widetilde{\Phi_{A}}(\rho)-\widetilde{\Phi_{B}}(\rho)}_1
  \geq \frac{2-b}{2}
\]
Let $\Xi=\widetilde{\Phi_{A}}-\widetilde{\Phi_{B}}$ and
$\equil{\lambda}(\Xi)$ be the equilibrium value defined before.
Finally, we have
\[
\left\{
  \begin{array}{ll}
  \equil{\lambda}(\Xi)\leq \frac{\sqrt{4-a^2}}{2}, & \hbox{$\snorm{\mathcal{Q}_0-\mathcal{Q}_1}_{\diamond} \geq a$;} \\
  \equil{\lambda}(\Xi)\geq \frac{2-b}{2} , & \hbox{$\snorm{\mathcal{Q}_0-\mathcal{Q}_1}_{\diamond} \leq b$.}
  \end{array}
\right.
\]
As we can see through the proof, the desired super operator $\Xi$ is
constructed explicitly from the input circuits $\Q_0,\Q_1$.
Moreover, every step in the construction only involves fundamental
operation of matrices. Due to the facts in Section~\ref{sec:NC}, we
are able to construct such $\Xi$ efficiently in parallel.
\end{proof}

%-----------------------------------------------------------------------------%
\section{Multiplicative Weights Update method for Computing Equilibrium Values}
\label{sec:MMW_PSPACE}
%-----------------------------------------------------------------------------%
The \emph{multiplicative weights update method} introduced in
Section \ref{sec:introduction} is a framework for algorithm design
(or meta-algorithm) that works as the one shown in
Fig~\ref{fig:mmw}. This kind of framework involves lots of technical
details and we refer the curious reader to the survey and the PhD
thesis mentioned in the introduction. However, for the sake of
completeness, we provide the main result which will be useful in our
proof. It should be noticed that $\{M^{(t)}\}$ is the freedom we
have in this framework.

\begin{theorem} \label{thm:mmw_main}
After $T$ rounds, the algorithm in Fig~\ref{fig:mmw} guarantees
that, for any $\rho^{\ast} \in \density{\X}$, we have
\begin{equation}\label{eqn:mmw_main}
    (1-\epsilon)\sum_{\geq 0} \ip{\rho^{(t)}}{M^{(t)}}+(1+\epsilon)\sum_{\leq 0}
    \ip{\rho^{(t)}}{M^{(t)}} \leq \ip{\rho^{\ast}}{\sum_{t=1}^T
    M^{(t)}} + \frac{lnN}{\epsilon}
\end{equation}
Here, the subscripts $\geq 0$ or $\leq 0$ in the summation refer to
the rounds $t$ where $0\leq M^{(t)} \leq \I$ or $-\I \leq M^{(t)}
\leq 0$ respectively.
\end{theorem}

Since in our consideration, it always holds that $0\leq M^{(t)} \leq
\I$. It suffices for our purpose to prove a simpler version of the
theorem ~\ref{thm:mmw_main} although the proof is almost the same as
the one for the general version.

\begin{figure}[t]
\noindent\hrulefill
\begin{mylist}{8mm}
\item[1.] Initialization: Pick a fixed $\varepsilon\leq \frac{1}{2}$,
and let $W^{(1)}=\I_{\X} \in \lin{\X}$, $N=\dim{\X}$.
\item[2.]
Repeat for each $t = 1,\ldots,T$:
\begin{mylist}{8mm}
\item[(a)] Let the density operator
$\rho^{(t)}=W^{(t)}/\tr{W^{(t)}}$
\item[(b)]  Observe the loss matrix $M^{(t)} \in \lin{\X}$ which satisfies $-\I_\X \leq M^{(t)}\leq 0$ or $0\leq M^{(t)}\leq \I_\X$,
 update the weight matrix as follows:
\[
   W^{(t+1)}=exp(-\varepsilon \sum_{\tau=1}^t M^{(\tau)})
\]
\end{mylist}
\end{mylist}
\noindent\hrulefill \caption{The Matrix Multiplicative Weights
Update method.} \label{fig:mmw}
\end{figure}

\begin{theorem} \label{thm:mmw_simple}
Assume $0\leq M^{(t)} \leq \I$ for all $t$, after $T$ rounds, the
algorithm in Fig~\ref{fig:mmw} guarantees that, for any $\rho^{\ast}
\in \density{\X}$, we have
\begin{equation}\label{eqn:mmw_simple}
     (1-\epsilon)\sum_{t=1}^T
    \ip{\rho^{(t)}}{M^{(t)}} \leq \ip{\rho^{\ast}}{\sum_{t=1}^T
    M^{(t)}} + \frac{lnN}{\epsilon}
\end{equation}
\end{theorem}
We put off the proof in the appendix part. It will be sufficient to
just remember this theorem in the first reading and skip the
details.

%-----------------------------------------------------------------------------%
\subsection{Algorithm for Computing Equilibrium Values}
%-----------------------------------------------------------------------------%

Using the multiplicative weight update method to compute some kind
of equilibrium values was known before, for instance the equilibrium
value of zero-sum game (an algorithm to compute this value can be
found in Kale's thesis~\cite{Kale07} and the
survey~\cite{AroraHK05b}. The reference for similar algorithms with
different purposes can be found in the survey~\cite{AroraHK05b}).
However, to compute the equilibrium value defined in our form, we
need to adapt the old idea to the new situation.

In order to compute the equilibrium value $\equil{\lambda}(\Xi)$, we
design an algorithm as shown in Fig~\ref{fig:minmax}. This algorithm
takes the descriptions of the two mixed-state quantum circuits as
input, and then compute the
$\Xi=\widetilde{\Phi_{A}}-\widetilde{\Phi_{B}}$ in \emph{Theorem
~\ref{thm:conversion}}, and output the equilibrium value
$\equil{\lambda}(\Xi)$ with precision $\delta$. Namely, the return
value $\lambda$ satisfies $|\lambda-\equil{\lambda}(\Xi)|\leq
\delta$.

Before proving the correctness of the algorithm, one might want to
compare the algorithms in both Fig~\ref{fig:mmw} and
Fig~\ref{fig:minmax}. The only change in our algorithm is that we
propose a way of computing $M^{(t)}$ for each round $t$. As we
mentioned before, $\{M^{(t)}\}$ is the freedom we have in this
framework. Different designs of $\{M^{(t)}\}$ can lead to different
applications of this framework. For instance, the primal-dual
approach for semidefinite programs in Kale's thesis~\cite{Kale07} is
an example of the design of $\{M^{(t)}\}$ that provides a good
application.

\begin{figure}[t]
\noindent\hrulefill
\begin{mylist}{8mm}
\item[1.]
Let $\varepsilon=\frac{\delta}{4}$ and $T=\ceil{\frac{16 \ln
N}{\delta^2}}$. Also let $W^{(1)}=\I_{\X}$, $N=\dim{(\X)}$.
\item[2.]
Repeat for each $t = 1,\ldots,T$:

\begin{mylist}{8mm}
\item[(a)]
Let $\rho^{(t)}=W^{(t)}/\tr{W^{(t)}}$ and compute the
$\Xi(\rho^{(t)})$. Then let $\Pi^{(t)}$ be the projection onto the
positive eigenspaces of $\Xi(\rho^{(t)})$.
\item[(b)]
Let $M^{(t)}=(\Xi^{\ast}(\Pi^{(t)})+\I_\X)/2$, and update the weight
matrix as follows:
\[
   W^{(t+1)}=exp(-\varepsilon \sum_{\tau=1}^t M^{(\tau)})
\]
\end{mylist}

\item[3.]
Return $\frac{1}{T} \sum_{t=1}^T \ip{\Pi^{(t)}}{\Xi(\rho^{(t)})}$ as
the approximation of $\equil{\lambda}(\Xi)$.
\end{mylist}
\noindent\hrulefill \caption{An algorithm that computes the
approximation $\equil{\lambda}(\Xi)$ with precision $\delta$. }
\label{fig:minmax}
\end{figure}

\begin{theorem} \label{thm:algorithm}
Using $T=\ceil{\frac{16 \ln N}{\delta^2}}$ rounds, the algorithm in
Fig~\ref{fig:minmax} returns the approximated value of
$\equil{\lambda}(\Xi)$ with precision $\delta$. Namely, we have the
return value $\lambda$ satisfying
\[
|\lambda-\equil{\lambda}(\Xi)|\leq \delta
\]
\end{theorem}

\begin{proof} First note that for any $\Pi^{(t)}$ computed during the process,
\[
   \forall \rho \in \density{\X}, \quad  |\ip{\rho}{\Xi^\ast(\Pi^{(t)})}|\leq
   1
\]
since $\Xi(\rho)$ is the difference between two density operators.
Thus, $M^{(t)}=(\Xi^{\ast}(\Pi^{(t)})+\I_{\X})/2$ satisfies  $0 \leq
M^{(t)} \leq \I_{\X}$.

 Then apply \emph{Theorem~\ref{thm:mmw_simple}}, we have,

\begin{equation}\label{eqn:mmw_main_apply}
 (1-\varepsilon)\sum_{\tau=1}^T \ip{\rho^{(\tau)}}{M^{(\tau)}} \leq
 \ip{\rho^{\ast}}{\sum_{\tau=1}^T M^{(\tau)}}+\frac{\ln N}{\varepsilon}
\end{equation}
for any density operator $\rho^{\ast} \in \density{\X}$. Substitute
$M^{(t)}=(\Xi^{\ast}(\Pi^{(t)})+\I_{\X})/2$ into Eq
[\ref{eqn:mmw_main_apply}] and divide both side by $T$, note that
$\ip{\rho^{(t)}}{M^{(t)}}\leq 1$, then we have
\begin{equation*} %\label{eqn:mmw_stepone}
\frac{1}{T} \sum_{\tau=1}^{T}
\ip{\rho^{(\tau)}}{\Xi^{\ast}(\Pi^{(\tau)})} \leq \frac{1}{T}
\ip{\rho^{\ast}}{\sum_{\tau=1}^T \Xi^{\ast}(\Pi^{(\tau)})} +
2\varepsilon +\frac{2\ln N}{\varepsilon T}
\end{equation*}
By choosing $\varepsilon=\frac{\delta}{4}$ and $T=\ceil{\frac{16 \ln
N}{\delta^2}}$, we have
\begin{equation}\label{eqn:mmw_steptwo}
\lambda = \frac{1}{T} \sum_{\tau=1}^{T}
\ip{\rho^{(\tau)}}{\Xi^{\ast}(\Pi^{(\tau)})} \leq \frac{1}{T}
\ip{\rho^{\ast}}{\sum_{\tau=1}^T \Xi^{\ast}(\Pi^{(\tau)})} + \delta
\end{equation}

In each step, $\Pi^{(t)}$ is returned as the solution to maximize $
  \ip{\Pi^{(t)}}{\Xi(\rho^{(t)})} $
for any fixed $\rho^{(t)}$. Due to the definition of the equilibrium
value in Eq [\ref{eqn:minmax_thm}], the equilibrium value
$\equil{\lambda}(\Xi) \leq \ip{\Pi^{(t)}}{\Xi(\rho^{(t)})}$ for any
returned $\Pi^{(t)}$.  On the other side, choose
$(\equil{\rho},\equil{\Pi})$ to be any \emph{equilibrium point} and
let $\rho^\ast=\equil{\rho}$, then we have
$\ip{\equil{\rho}}{\Xi^{\ast}(\Pi^{(t)})}\leq \equil{\lambda}(\Xi)$.
Using inequality [\ref{eqn:mmw_steptwo}], we have
\begin{equation}\label{eqn:mmw_stepthree}
\equil{\lambda}(\Xi)\leq \lambda \leq \frac{1}{T}
\ip{\equil{\rho}}{\sum_{\tau=1}^T \Xi^{\ast}(\Pi^{(\tau)})} + \delta
\leq \equil{\lambda}(\Xi)+\delta
\end{equation}
which completes the proof
\end{proof}

To distinguish between the two promises in
\emph{Theorem~\ref{thm:conversion}}, we let $\delta=0.2$ and make
use of the approximated equilibrium value returned in the algorithm.
If the value is closer to $0.95$, then it is the case that
$\snorm{\Phi}_\diamond\leq 0.1$. Otherwise, it is the case that
$\snorm{\Phi}_\diamond \geq 1.9$. Thus, we solve the promise QCD
problem in this way.

%-----------------------------------------------------------------------------%
\subsection{Simulation by bounded-depth Boolean circuits}
\label{sec:NC}
%-----------------------------------------------------------------------------%
We denote by \class{NC} the class of promise problems computed by
the logarithmic-space uniform Boolean circuits with poly-logarithmic
depth. Furthermore, we denote by \class{NC(poly)} the class of
promise problems computed by the polynomial-space uniform Boolean
circuits with polynomial depth. Since it holds that
\class{NC(poly)}=\class{PSPACE}, thus in order to simulate the
algorithm above in PSPACE, it suffices to prove that we can simulate
the algorithm in \class{NC(poly)}.

There are a few facts about these classes which are useful in our
discussion. The first fact is the functions in these classes compose
nicely. It is clear that if $f \in$ \class{NC(poly)} and $g \in$
\class{NC}, then their composition $g\circ f$ is in
\class{NC(poly)}, which follows from the most obvious way of
composing the families of circuits. Another useful fact is that many
computations involving matrices can be performed by NC algorithms
(Please refer to the survey~\cite{vzGathen93} which describes NC
algorithms for these tasks). Especially, we will make use of the
fact that matrix exponentials and positive eigenspace projections
can be approximated to high precision in NC. A more careful
treatment on those issues can be found in
Appendix~\ref{sec:precision_issue}.

Since we are able to perform matrix operations with sufficient
accuracy in \class{NC}, it remains to show the ability to compose
all the operations in \class{NC(poly)} and thus in \class{PSPACE}.
Precisely,

\begin{theorem} \label{thm:NC_simulation}
The algorithm shown above can solve QCD problems in
\class{NC(poly)}, and thus in \class{PSPACE}.
\end{theorem}

\begin{proof} To simulate the algorithm, it suffices to compose the
following families of Boolean circuits.
\begin{enumerate}
  \item A family of Boolean circuits that output the representation
  of the quantum channel $\Xi$ (in Theorem~\ref{thm:conversion}) generated from the input $x$, namely, the
  descriptions of two mixed quantum circuits.
  \item Follow the algorithm in Figure~\ref{fig:minmax}. Compose all
  the operations in each iteration. Consider the fact that fundamental matrix
  operations
  can be done in \class{NC} and the number of iterations $T=\ceil{\frac{16 \ln
N}{\delta^2}}$ is polynomial in the size of $x$ since $\delta$ is a
constant and $N$ is exponential in the size of $x$.
  \item The circuits to distinguish between the two promises by making use of the
  value returned in the circuits above.
\end{enumerate}
The first family is easily done in \class{NC(poly)}, by computing
the product of a polynomial number of exponential-size matrices
which corresponds to the mixed quantum circuits. The second family
is in \class{NC} by composing polynomial number of \class{NC}
circuits. The third one is obviously in \class{NC}. The whole
process is in \class{NC(poly)} by composing the \class{NC(poly)} and
\class{NC} circuits above, and thus in \class{PSPACE}.
\end{proof}

It follows from the \emph{Theorem~\ref{thm:NC_simulation}} that we
can solve $QCD_{1.9,0.1}$ problems in \class{PSPACE}. Since
$QCD_{1.9,0.1}$ is QIP-Complete problem, and any polynomial
reduction to that problem can be easily done in \class{NC(poly)} by
computing the product of a polynomial number of exponential-size
matrices, thus we have \class{QIP}$\subseteq$\class{PSPACE}.
Combining with the known result
\class{PSPACE}=\class{IP}$\subseteq$\class{QIP}, we have,

\begin{corollary}
  \class{QIP}=\class{PSPACE}
\end{corollary}

We notice that all the proof above so far is based on the assumption
that all the matrix operations can be simulated exactly. However in
practice, we will need to truncate the precision to some place for
some operations to be performed. Fortunately, this won't be an
essential obstacle for the implementation of the algorithm. As
mentioned in the introduction, all those issues will be handled in
Appendix~\ref{sec:precision_issue} without any change of the main
result.

%-----------------------------------------------------------------------------%
\section{Algorithm for computing the diamond norm} \label{sec:diamond}
%-----------------------------------------------------------------------------%
Now it is our turn to discuss the computation of the diamond norm
for a special class of quantum channels. Consider the following
promised version of the problem first.

\begin{definition} [Promised Diamond Norm Problem]
Given the classical description of any two admissible quantum
channels $\Q_0, \Q_1$, the promised diamond norm problem
$\text{PDN}(\Q_0,\Q_1,a,b)$ is asked to distinguish between the two
cases , namely whether $\snorm{\Q_0-\Q_1}_\diamond\geq a$ or
$\snorm{\Q_0-\Q_1}_\diamond \leq b$ where $a-b \in
\Omega(1/\text{poly-log})$.
\end{definition}

As mentioned in the introduction part, Theorem~\ref{thm:conversion}
explicitly makes one conversion from diamond norms to equilibrium
values that perverse the promised gap. It is easy to see that
Theorem~\ref{thm:conversion} works for any admissible channels
$\Q_0, \Q_1$. Furthermore, this conversion can be done efficiently
as long as the classical descriptions of $\Q_0, \Q_1$ are given.
Thus, by combing the results in Theorem~\ref{thm:algorithm}, one can
easily solve the promised version of diamond norm problems. Note
that the input size changes to be the size of the matrix
representing the channel now. %Precisely,

\begin{prop} \label{prop:simple_case}
There is a \class{NC} algorithm that solves
$\text{PDN}(\Q_0,\Q_1,a,b)$ where $a^2-(4b-b^2) \in
\Omega(1/\text{poly-log})$.
\end{prop}

\begin{proof}
This is a direct consequence when one combines the result of
Theorem~\ref{thm:conversion} and Theorem~\ref{thm:algorithm}. Given
the promise, by Theorem~\ref{thm:conversion}, one can efficiently
compute the equilibrium value $\equil{\lambda}(\Xi)$ whose value is
either no more than $\frac{\sqrt{4-a^2}}{2}$ or no less than
$\frac{2-b}{2}$. Hence if the difference $a^2-(4b-b^2) \in
\Omega(1/\text{polylog})$, one can use the algorithm in
Fig~\ref{fig:minmax} to calculate $\equil{\lambda}(\Xi)$ efficiently
in parallel to sufficient precision in order to distinguish between
those two cases. The \class{NC} algorithm follows directly when one
composes the circuits for each step.
\end{proof}

The only undesired thing of this algorithm is we can only solve the
problem when the condition $a^2-(4b-b^2) \in
\Omega(1/\text{polylog})$. This constraint makes it impossible for
our algorithm to work for the whole range $0\leq b \leq a \leq 2$.
Evidences (implicitly in~\cite{RosgenW05}, e.g. Theorem 4.3) also
demonstrate that simply repeating the channels for many times, like
considering the alternative channels $(\Q_0-\Q_1)^{\otimes N}$ or
$\Q_0^{\otimes N}-\Q_1^{\otimes N}$ for some $N$, doesn't work
either.

Fortunately, there is one conceptually easy but technically detoured
approach to amplify the gap in general. Particularly, we will make
use of some known properties of the quantum interactive proof
systems and abuse them for a different purpose. Intuitively, we
treat any quantum interactive proof protocol (assume the input is
fixed) as a promised problem where the acceptance probability is
either at least the completeness probability or no more than the
soundness probability. Then we will convert the promised diamond
norm problem into such a promised problem of one specific quantum
interactive proof protocol. The crucial observation is the
parallelization, amplification lemmas together with the complete
problem results about quantum interactive proof systems can be
exploited to amplify arbitrary gap of any general promised diamond
norm problem and convert it to a new diamond norm
problem\footnote{Actually, it suffices to convert the original
problem to \emph{close images} problem. However, for the simplicity
of description, we choose \emph{QCD} problem instead.} that can be
solved by Proposition~\ref{prop:simple_case}.

Let us demonstrate this approach with full detail. More importantly,
we will show such conversion can also be computed efficiently. The
latter one is due to the fact the parallelization and amplification
procedures are constructed explicitly in~\cite{KitaevW00}. Any
protocol $\mathbf{P}$\footnote{Precisely, a single prover quantum
interactive proof protocol.} with completeness $a$ and soundness $b$
will be denoted by $\mathbf{P}(a,b)$. Recall that any such protocol
$\mathbf{P}(a,b)$ is treated as a promise problem where the
acceptance probability is either at least $a$ or at most $b$. The
promised diamond norm problem $\text{PDN}(\Q_0,\Q_1,a,b)$ can thus
be converted to the following protocol directly.

\begin{definition}
The protocol $\mathcal{P}_\diamond[\Q_0,\Q_1]$:
\begin{enumerate}
  \item The verifier receives some quantum state $\rho$ from the
  prover.
  \item The verifier selects $\{0,1\}$ uniformly and applies $\Q_i$
  to $\rho$ and sends the result to the prover.
  \item The verifier receives some $j$ from the prover, accepts if
  $i=j$ and rejects otherwise.
\end{enumerate}
\end{definition}

The protocol is almost identical to the Protocol 3.2
in~\cite{RosgenW05}. The only difference is the verifier needs to
perform arbitrary admissible quantum channel $\Q_i$. It is not
possible in general when the verifier's power is polynomial time
bounded since arbitrary $\Q_i$ might need huge number of gates to
simulate. However, it won't be an issue for us since we treat such
protocol as a promised problem without its original meaning and
$\Q_0,\Q_1$'s description is already given. It follows immediately
from~\cite{RosgenW05} that the protocol
$\mathbf{P}_\diamond[\Q_0,\Q_1]$ has completeness $1/2+a/4$ and
soundness $1/2+b/4$ given the promise that either
$\snorm{\Q_0-\Q_1}_\diamond\geq a$ or $\snorm{\Q_0-\Q_1}_\diamond
\leq b$.

% However, this mild assumption won't affect our results because of
%the following two points.
%Firstly, the constructions and results about the parallelization and
%amplification of quantum interactive proof protocol wont't change
%with this additional assumption. Secondly, since $\Q_0,\Q_1$'s
%classical descriptions are given as input, we don't need to worry
%about the computation of their description. At the same time, the
%whole circuit's description can be computed efficiently in parallel
%because other circuits except $Q_i$ are only of polynomial size.

The parallelization and amplification lemmas~\cite{KitaevW00} can
then be reinterpreted as a way to convert any protocol
$\mathbf{P}(c,d)$ to some protocol $\mathbf{P}'(c',d')$ with desired
$c'$ and $d'$. This conversion can be efficiently computed when the
gap between $c$ and $d$ is appropriate. For the promised diamond
norm problem, we starts with some protocol
$\mathbf{P}_\diamond(a,b)$  where the gap between $a$ and $b$ is at
least inverse poly-logarithm and converts it to some protocol
$\mathbf{P}'_\diamond(1,1/2)$. Moreover, such conversion can be done
in \class{NC}. Secondly, because $QCD_{1.9,0.1}$ is QIP-Complete
problem, we can convert the protocol $\mathbf{P}'_\diamond(1,1/2)$
again to a new promised diamond norm problem $\text{PDN}(\Q_0',
\Q_1', 1.9, 0.1)$ where such conversion is implicitly inside the
proof of QIP-Completeness of the $QCD_{1.9,0.1}$~\cite{RosgenW05}
and the new channels $\Q_0', \Q_1'$ can be computed in \class{NC} as
well. Finally, we will invoke the algorithm in
Proposition~\ref{prop:simple_case} to solve the new problem. To sum
up,

\begin{figure}[t]
\noindent\hrulefill
\begin{mylist}{8mm}
\item[1.] If $a^2-(4b-b^2) \in
\Omega(1/\text{polylog})$, use the algorithm in
Proposition~\ref{prop:simple_case} to solve it directly. Otherwise,
continue to next step.
\item[2.] Convert the original problem to
$\mathbf{P}_\diamond[\Q_0,\Q_1](1/2+a/4,1/2+b/4)$.
\item[3.] According to the parallelization lemma and amplification
lemma~\cite{KitaevW00}, convert the protocol
$\mathbf{P}_\diamond[\Q_0,\Q_1](1/2+a/4,1/2+b/4)$ to
$\mathbf{P}'_\diamond(1,1/2)$. According to the construction
implicit in the proof of the QIP-Completeness of the problem
$QCD_{1.9,0.1}$, convert the protocol $\mathbf{P}'_\diamond(1,1/2)$
to a new promised diamond norm problem $\text{PDN}(\Q_0', \Q_1',
1.9, 0.1)$.
\item[4.] Use the algorithm in Proposition~\ref{prop:simple_case} to
solve the new problem $\text{PDN}(\Q_0', \Q_1', 1.9, 0.1)$ and
return the answer.
\end{mylist}
\noindent\hrulefill \caption{Algorithm for the
$\text{PDN}(\Q_0,\Q_1,a,b)$ problem.} \label{fig:pdn}
\end{figure}

\begin{prop}
There is a \class{NC} algorithm shown in Fig.~\ref{fig:pdn} that
solves general $\text{PDN}(\Q_0,\Q_1,a,b)$ problems.
\end{prop}

As a standard technique, an algorithm for the promised version of
problems can be used as a subroutine to solve the general problems
via binary search when the range of the possible results is bounded.
In our case, the diamond norm between any two admissible quantum
channels is bounded between 0 and 2. Hence, by recursively calling
the subroutine in Fig.~\ref{fig:pdn} $O(\text{poly-log})$ times, one
can compute the diamond norm with $\Omega(1/\text{poly-log})$
precision.

\begin{corollary}
Given the classical description of any two admissible quantum
channels $\Q_0, \Q_1$, the diamond norm of their difference
$\snorm{\Q_0-\Q_1}_\diamond$ can be approximated in \class{NC} with
inverse poly-logarithm precision.
\end{corollary}

%First, to say, the result previously already imply that the promise
%version of computing the diamond norm in this situation can be
%solved correctly. But there are some details, replace some steps in
%the previous proof. Resettle the size and the precision.

%Second, question how to get to the good amplification? Demonstrate
%the difficulty if directly parallel the channel , show evidence from
%RW05. Then also show how to deal with that by using the QIP
%protocol, QIP amplification and conversion to complete problem to
%solve the problem in NC(poly)

%-----------------------------------------------------------------------------%
\section{Summary} \label{sec:summary}
%-----------------------------------------------------------------------------%
In this paper, we provide an alterative proof for
\class{QIP}=\class{PSPACE} which starts from one
\class{QIP}-Complete problem that computes the diamond norm between
two quantum admissible channels. The key observation here is to
convert the computation of the diamond norm to the computation of an
equilibrium value. The later problem turns out to be a more
structured problem and has a good solution in \class{NC(poly)} and
thus in \class{PSPACE}. Besides reducing from the QCD problem, we
could also reduce from the very first \class{QIP}-Complete problem
\emph{close images}~\cite{KitaevW00} or the protocol to simulate
\class{QIP} with competing provers in~\cite{GutoskiW05}. Both
reductions will lead to the similar equilibrium values to the one in
this paper. The technique of computing the equilibrium values in
this paper can then be applied directly and lead to the same result.

The multiplicative weights update method in our proof to solve the
equilibrium value problem can be generalized to solve a class of
such equilibrium value problems. Particularly, for any density
operator set $\density{\X}$ and another convex compact set $\Gamma$,
the following general equilibrium value problem
\[
  \equil{\lambda}(\Phi)=\min_{\rho \in \density{\X}} \max_{\sigma \in \Gamma}
\ip{\sigma}{\Phi(\rho)} = \max_{\sigma \in \Gamma} \min_{\rho \in
\density{\X}} \ip{\Phi^{\ast}(\sigma)}{\rho}
\]
can be solved efficiently in \class{NC} by the same algorithm in our
paper if a good approximation algorithm to compute $\max_{\sigma \in
\Gamma} \ip{\sigma}{\Phi(\rho^\ast)}$ given $\rho^\ast$ is available
in \class{NC} and $\max_{\rho \in \density{\X}} \max_{\sigma \in
\Gamma} |\ip{\sigma}{\Phi(\rho)}|$ is bounded by some poly-logarithm
function.

One big open problem is to investigate to what extent the technique
in this paper can be used to solve general equilibrium value
problems. As mentioned in~\cite{JainJUW09}, it is impossible to
solve arbitrary SDPs in parallel unless \class{NC}=\class{P}. It
might be the same case for the general equilibrium value problems. A
recent effort~\cite{Wu10,GutoskiW10} made some progress on the
general form of the equilibrium value that can be solved by similar
techniques and the connection between the equilibrium value problems
and semidefinite programming problems. Moreover, since the first
release of this paper, some deeper knowledge of variants of quantum
interactive proof system is obtained. Particularly, the main open
problem in~\cite{JainJUW09}, namely whether
\class{QRG(2)}=\class{PSPACE}, is resolved~\cite{GutoskiW10} with
positive answer. The class \class{QRG(2)} contains all problems
which can be recognized by two-turn (i.e, one-round) quantum
refereed games. The classical analogue of this class is known to
coincide with \class{PSPACE}~\cite{FeigeK97}.

Another open problem is how to extend the connection we build here
between the computation of diamond norms and the computation of
equilibrium values. One might hope to obtain efficient parallel
algorithm for calculating diamond norms of any quantum channel.

%-----------------------------------------------------------------------------%
\subsection*{Acknowledgments}
%-----------------------------------------------------------------------------%
The author is grateful to Zhengfeng Ji, Yaoyun Shi and John Watrous
for helpful discussions. The author also wants to thank the
hospitality and invaluable guidance of John Watrous when the author
was visiting the Institute for Quantum Computing, University of
Waterloo. The research was completed during this visit and was
supported by the Canadian Institute for Advanced Research (CIFAR).
The author also wants to thank Rahul Jain, Zhengfeng Ji and
anonymous reviewers for helpful comments on the manuscript. The
author's research was also partially supported by US NSF under the
grants 0347078 and 0622033.

\appendix
\section{Proof of Theorem \ref{thm:mmw_simple}}
Assume $0\leq M^{(t)} \leq \I$ for all $t$, after $T$ rounds, the
algorithm in Fig~\ref{fig:mmw} guarantees that, for any $\rho^{\ast}
\in \density{\X}$, we have
\begin{equation*} %\label{eqn:mmw_simple}
     (1-\epsilon)\sum_{t=1}^T
    \ip{\rho^{(t)}}{M^{(t)}} \leq \ip{\rho^{\ast}}{\sum_{t=1}^T
    M^{(t)}} + \frac{lnN}{\epsilon}
\end{equation*}

\begin{proof}
It is easy to see that all $W^{(t)} \in \pos{\X}$. Observe that, for
$t=1, \cdots, T$,

\begin{align*}
\tr(W^{(t+1)}) & = \tr\left[
  \exp\(-\varepsilon\sum_{\tau=1}^t M^{(\tau)}\)
  \right]
  \leq
\tr\left[\exp\(-\varepsilon\sum_{\tau=1}^{t-1}
M^{(\tau)}\)\exp\(-\varepsilon M^{(t)}\)
  \right]\\
& = \tr\left[W^{(t)} \exp\(-\varepsilon M^{(t)}\) \right] \leq \tr
\left[W^{(t)}(\I_{\X}-\varepsilon'M^{(t)})
  \right] \\
& = \ip{W^{(t)}}{\I_{\X}-\varepsilon'M^{(t)}} =\tr
\left[W^{(t)}\right](1-\varepsilon'\ip{\rho^{(t)}}{M^{(t)}})\\
& \leq \tr \left[W^{(t)}\right] \exp (-\varepsilon'
\ip{\rho^{(t)}}{M^{(t)}})
\end{align*}
The inequality in the first line is due to Golden-Thompson
inequality~\cite{Bhatia97}. The second inequality is due to
Lemma~\ref{lem:inequality} where $\varepsilon'=1-e^{-\varepsilon}$.
The third line is obtained by substituting
$\rho^{(t)}=W^{(t)}/\tr{W^{(t)}}$. The final inequality is obtained
by noticing that $1-\varepsilon'x \leq e^{-\varepsilon'x}, \forall x
\in [0,1]$ and $\ip{\rho^{(t)}}{M^{(t)}} \in [0,1]$.

If we repeat the process above, by induction as well as the fact
$W^{(1)}=\I_{\X}$, we have:
\[
  \tr \left[ W^{(T+1)}\right] \leq N \exp(-\varepsilon'
  \sum_{\tau=1}^T \ip{\rho^{(\tau)}}{M^{(\tau)}})
\]
where $N=\dim{(\X)}$ as defined above. On the other hand, we have
\[
  \tr \left[ W^{(T+1)}\right]= \tr \left[ \exp(-\varepsilon
  \sum_{\tau=1}^T M^{(\tau)}) \right] \geq \exp (-\varepsilon
  \lambda_N (\sum_{\tau=1}^T M^{(\tau)}))
\]
The last inequality holds because $\tr(e^A)=\sum_{i=1}^{\dim(\A)}
e^{\lambda_i(A)} \geq e^{\lambda_1(A)}$. Thus, we conclude that
\[
  \exp(-\varepsilon\lambda_N(\sum_{\tau=1}^T M^{(t)})) \leq N \exp(-\epsilon'
 \sum_{\tau=1}^T \ip{\rho^{(t)}}{M^{(t)}})
\]
Since $\lambda_N$ is the minimum eigenvalue, for any density
operator $\rho^{\ast} \in \density{\X}$, we have
$\lambda_N(\sum_{\tau=1}^T M^{(t)})\leq
\ip{\rho^\ast}{\sum_{\tau=1}^T M^{(t)}}$. Take the logarithms of the
both side and simplify as well as notice the fact $\varepsilon'$
obtained by Lemma~\ref{lem:inequality} has the property that
$\varepsilon'\geq \varepsilon(1-\varepsilon)$, then we have:
\begin{equation*}%\label{eqn:mmw_main}
 (1-\varepsilon)\sum_{\tau=1}^T \ip{\rho^{(\tau)}}{M^{(\tau)}} \leq
 \ip{\rho^{\ast}}{\sum_{\tau=1}^T M^{(\tau)}}+\frac{\ln N}{\varepsilon}
\end{equation*}
for any density operator $\rho^{\ast} \in \density{\X}$ as required.
\end{proof}

 For the sake of completeness, we prove the
lemma we used in the proof below.
\begin{lemma} \label{lem:inequality}
For any $0\leq \varepsilon \leq 1/2$, let
$\varepsilon'=1-e^{-\varepsilon}$. Then we have the following matrix
inequality, for any $0\leq M \leq \I$,
\[
  \exp (- \varepsilon M) \leq \I-\varepsilon' M
\]
Moreover, we have $\varepsilon'\geq \varepsilon(1-\varepsilon)$.
\end{lemma}

\begin{proof}
It is easy to verify that for any $\varepsilon \in [0,1/2]$, we have
\[
 f(x) \defeq{} \exp ( -\varepsilon x ) \leq g(x)\defeq{} 1- \varepsilon'x  \quad \quad \text{if} \quad x \in
 [0,1]
\]
Since $0\leq M \leq \I$, let $M=UDU^{\dagger}$ be the diagonlization
of $M$. Then,$f(M)-g(M)=U(f(D)-g(D))U^{\dagger}$. Since $D$ is a
diagonal matrix of which every diagonal entry contains one
eigenvalue of $M$ and thus is in $[0,1]$, then $f(D)-g(D)$ is a
diagonal matrix with non-positive diagonal due to the inequality
above. Thus, $f(M)-g(M)$ is a negative semidefinite matrix and hence
$f(M) \leq g(M)$, namely, $  \exp (- \varepsilon M) \leq
\I-\varepsilon' M$. It is also easy to verify that $\varepsilon'\geq
\varepsilon(1-\varepsilon)$.
\end{proof}

\section{Comments on precision issues}
\label{sec:precision_issue}

The analysis made in the main part of this paper has assumed that
all computations performed by the algorithm are exact. However, in
order to implement our algorithm, some step of the computations must
be approximate. Particularly, the computation of the positive
eigenspace projections and the matrix exponentials will need to be
approximate. An elaborated analysis on these issues can be found in
~\cite{JainJUW09,JainW09}. We will basically follow that type of
analysis and provide a brief sketch of the analysis to the specific
problem in our paper.

First, it must be made clear which part of the algorithm can be made
exact and which part must be made approximate. We will use the same
convention of storing complex numbers as the one
in~\cite{JainJUW09}. Once the input $x$ is given and stored in
memory, all elementary matrix operations (in this case: addition,
multiplication, and computation of the trace or partial trace) can
be implemented exactly in \class{NC}~\cite{vzGathen93}. However, the
matrix exponentials and positive eigenspace projection cannot be
exact since these operations will generate irrational numbers and
the precision must be truncated at somewhere. Fortunately,
Watrous\emph{ et al.}~\cite{JainJUW09} provided a way to approximate
those two operations to high precision in \class{NC}. Precisely,

\begin{fact} \label{fact:exp_NC}
Given an $n \times n$ matrix $M$ (whose operator norm bounded by
$k$) and a positive rational number $\eta$, the computation of $n
\times n$ matrix $X$ such that $\snorm{exp(M)-X} < \eta$ can be done
in \class{NC}.
\end{fact}

\begin{fact} \label{fact:pos_NC}
Given an $n \times n$ Hermitian matrix $H$ and a positive rational
number $\eta$, the computation of an $n \times n$ positive
semidefinite matrix $\Delta \leq \I$ such that
$\snorm{\Delta-\Lambda} <\eta$ for $\Lambda$ being the projection
operator onto the positive eigenspace of $H$ can be done in
\class{NC}.
\end{fact}

Before we move on to the analysis of the precision issue, it helps
to introduce the following convention. We will represent the actual
matrices generated during the algorithm by placing a tilde over the
variables that represent the idealized values. As we discussed
above, there are mainly two types of operations where the accuracy
will be lost. Further investigation tells us that the matrix
exponentials are always necessary to the multiplicative weight
update method while the positive eigenspace projections are special
for our application. For the generality of the analysis, we will
first discuss what the general form of Theorem~\ref{thm:mmw_simple}
is when the computation is only approximate.

Consider the scheme in Fig~\ref{fig:mmw} and keep the notation
convention in mind. The $\tilde{\rho}^{(t)}$ will be the actual
generated density operator for round $t$ and
$W^{(t+1)}=exp(-\epsilon\sum_{\tau=1}^t \tilde{M}^{(t)})$. The
latter one is exact simply because $W^{(t)}$ is only a notation and
not stored in the memory at all. Fact~\ref{fact:exp_NC} implies that
$\snorm{\tilde{\rho}^{(t)}-W^{(t)}/\tr W^{(t)}} < \delta_1/N$ for
every $t$ where $\delta_1$ is some constant for our purpose. The
situation for $\tilde{M}^{(t)}$ is tricky in the sense that there is
no idealized value for $M^{(t)}$ to have in the general scheme. By
going through the proof of Theorem~\ref{thm:mmw_simple} again, we
can easily obtain the following fact.

\begin{fact}
If the computation can only be performed approximately, the
inequality in Theorem~\ref{thm:mmw_simple} becomes
\[
     (1-\epsilon)\sum_{t=1}^T
    \ip{\tilde{\rho}^{(t)}}{\tilde{M}^{(t)}} \leq \ip{\rho^{\ast}}{\sum_{t=1}^T
    \tilde{M}^{(t)}} + \frac{lnN}{\epsilon} + \frac{1}{2} T \delta_1
\]
\end{fact}

Now consider the concrete $\tilde{M}^{(t)}$ in Fig~\ref{fig:minmax}.
By making use of the fact above we can repeat almost all the steps
in the proof of Theorem~\ref{thm:algorithm}. The only change is to
replace the Equation~[\ref{eqn:mmw_steptwo}] by
\[
\lambda = \frac{1}{T} \sum_{\tau=1}^{T}
\ip{\tilde{\rho}^{(\tau)}}{\Xi^{\ast}(\tilde{\Pi}^{(\tau)})} \leq
\frac{1}{T} \ip{\rho^{\ast}}{\sum_{\tau=1}^T
\Xi^{\ast}(\tilde{\Pi}^{(\tau)})} + \delta + \delta_1
\]

By Fact~\ref{fact:pos_NC}, we have
$\snorm{\tilde{\Pi}^{(t)}-\Pi^{(t)}} < \delta_1/N$ where the
$\Pi^{(t)}$ is the projection onto the positive eigenspace of
$\Xi(\tilde{\rho}^{(t)})$. Please note that $\Pi^{(t)}$ here is not
its idealized value when everything is exact but rather the exact
value given the approximate $\tilde{\rho}^{(t)}$. The rest part of
the proof follows similarly. Finally, we will get
\[
\equil{\lambda}(\Xi)-\delta_1 \leq \lambda=\frac{1}{T}
\sum_{\tau=1}^{T}
\ip{\tilde{\rho}^{(\tau)}}{\Xi^{\ast}(\tilde{\Pi}^{(\tau)})} \leq
\equil{\lambda}(\Xi)+\delta+\delta_1
\]

Since our target is to distinguish between two promises with
constant gap, we can choose $\delta_1$ to be any sufficiently small
constant. In this way, the precision issues are handled.

\bibliographystyle{alpha}

\begin{thebibliography}{BOGKW88}

\bibitem[AKN98]{AKN98}
D.~Aharonov, A.~Kitaev, and N.~Nisan.
\newblock Quantum circuits with mixed states.
\newblock In {\em Proceedings of the 30th Annual ACM Symposium on Theory of
Computing}, pages 20-30, 1998

\bibitem[AHK05a]{AroraHK05}
S.~Arora, E.~Hazan, and S.~Kale.
\newblock Fast algorithms for approximate semidefinite programming using the
  multiplicative weights update method.
\newblock In {\em Proceedings of the 46th Annual IEEE Symposium on Foundations
  of Computer Science}, pages 339--348, 2005.

\bibitem[AHK05b]{AroraHK05b}
S.~Arora, E.~Hazan, and S.~Kale.
\newblock The multiplicative weights update method: a meta algorithm and applications.
2005

\bibitem[AK07]{AroraK07}
S.~Arora and S.~Kale.
\newblock A combinatorial, primal-dual approach to semidefinite programs.
\newblock In {\em Proceedings of the Thirty-Ninth Annual ACM Symposium on
  Theory of Computing}, pages 227--236, 2007.

\bibitem[Bab85]{Babai85}
L.~Babai.
\newblock Trading group theory for randomness.
\newblock In {\em Proceedings of the 17th Annual ACM Symposium on Theory of
  Computing}, pages 421--429, 1985.

\bibitem[BATS09]{BATS09}
A.~Ben-Aroya, A.~Ta-Shma.
\newblock On the complexity of approximating the diamond norm,
\newblock arXiv 0902.3397, 2009.

%\bibitem[BCP83]{BorodinCP83}
%A.~Borodin, S.~Cook, and N.~Pippenger.
%\newblock Parallel computation for well-endowed rings and space-bounded
%  probabilistic machines.
%\newblock {\em Information and Control}, 58:113--136, 1983.
%
%\bibitem[BGH82]{BorodinGH82}
%A.~Borodin, J.~von~zur Gathen, and J.~Hopcroft.
%\newblock Fast parallel matrix and {GCD} computations.
%\newblock In {\em Proceedings of the 23rd Annual IEEE Symposium on Foundations
%  of Computer Science}, pages 65--71, 1982.

\bibitem[Bha97]{Bhatia97}
R.~Bhatia.
\newblock {\em Matrix Analysis}.
\newblock Springer, 1997.

\bibitem[BM88]{BabaiM88}
L.~Babai and S.~Moran.
\newblock {A}rthur-{M}erlin games: a randomized proof system, and a hierarchy
  of complexity classes.
\newblock {\em Journal of Computer and System Sciences}, 36(2):254--276, 1988.


%\bibitem[BOFKT86]{BenOrFKT86}
%M.~Ben-Or, E.~Feig, D.~Kozen, and P.~Tiwari.
%\newblock A fast parallel algorithm for determining all roots of a polynomial
%  with real roots.
%\newblock In {\em Proceedings of the 18th Annual ACM Symposium on Theory of
%  Computing}, pages 340--349, 1986.

\bibitem[BOGKW88]{Ben-OrGKW88}
M.~Ben-Or, S.~Goldwasser, J.~Kilian, and A.~Wigderson.
\newblock Multi-prover interactive proofs: how to remove intractability
  assumptions.
\newblock In {\em Proceedings of the 20th Annual ACM Symposium on Theory of
  Computing}, pages 113--131, 1988.

\bibitem[Bor77]{Borodin77}
A.~Borodin.
\newblock On relating time and space to size and depth.
\newblock {\em SIAM Journal on Computing}, 6:733--744, 1977.

%\bibitem[Csa76]{Csanky76}
%L.~Csanky.
%\newblock Fast parallel matrix inversion algorithms.
%\newblock {\em SIAM Journal on Computing}, 5(4):618--623, 1976.

\bibitem[Fan53]{Fan53}
K.~Fan.
\newblock Minimax theorems.
\newblock {\em Proceedings of the National Academy of Sciences}, 39:42--47, 1953.


\bibitem[FK97]{FeigeK97}
U.~Feige and J.~Kilian.
\newblock Making games short.
\newblock In {\em Proceedings of the 29th Annual ACM Symposium on Theory of
  Computing}, pages 506--516, 1997.

\bibitem[Gat93]{vzGathen93}
J.~von~zur Gathen.
\newblock Parallel linear algebra.
\newblock In J.~Reif, editor, {\em Synthesis of Parallel Algorithms},
  chapter~13. Morgan Kaufmann Publishers, Inc., 1993.

\bibitem[GMR85]{GoldwasserMR85}
S.~Goldwasser, S.~Micali, and C.~Rackoff.
\newblock The knowledge complexity of interactive proof systems.
\newblock In {\em Proceedings of the 17th Annual ACM Symposium on Theory of
  Computing}, pages 291--304, 1985.

\bibitem[GMW91]{GoldreichMW91}
O.~Goldreich, S.~Micali, and A.~Wigderson.
\newblock Proofs that yield nothing but their validity or all languages in {NP}
  have zero-knowledge proof systems.
\newblock {\em Journal of the ACM}, 38(1):691--729, 1991.

\bibitem[GS89]{GoldwasserS89}
S.~Goldwasser and M.~Sipser.
\newblock Private coins versus public coins in interactive proof systems.
\newblock In S.~Micali, editor, {\em Randomness and Computation}, volume~5 of
  {\em Advances in Computing Research}, pages 73--90. JAI Press, 1989.

\bibitem[Gut05]{Gutoski05}
G.~Gutoski
\newblock Upper bounds for quantum interactive proofs with competing
provers.
\newblock In {\em Proceedings of the 20th Annual IEEE Conference on Computational Complexity}, pages 334--343, 2005.

\bibitem[GW05]{GutoskiW05}
G.~Gutoski and J.~Watrous
\newblock Quantum interactive proofs with competing
provers.
\newblock In {\em Proceedings of the 22th Symposium on Theoretical Aspects of Computer Science},
  volume 3404 of {\em Lecture Notes in Computer Science}, pages 605--616. Springer 2005.

\bibitem[GW07]{GutoskiW07}
G.~Gutoski and J.~Watrous
\newblock Toward a general theory of quantum games.
\newblock In {\em Proceedings of the 39th ACM Symposium on Theory of Computing}, pages 565--574, 2007.

\bibitem[GW10]{GutoskiW10}
G.~Gutoski and X.~Wu
\newblock Short quantum games charaterize PSPACE.
\newblock available at arXiv: 1011.2787, 2010.

\bibitem[HKSZ08]{HallgrenKSZ08}
S.~Hallgren, A.~Kolla, P.~Sen, and S.~Zhang.
\newblock Making classical honest verifier zero knowledge protocols secure
  against quantum attacks.
\newblock In {\em Proceedings of the 35th International Colloquium on Automata,
  Languages and Programming}, volume 5126 of {\em Lecture Notes in Computer
  Science}, pages 592--603. Springer, 2008.

\bibitem[JJUW09]{JainJUW09}
R. Jain, Z. Ji, S. Upadhyay, and J. Watrous.
\newblock QIP = PSPACE.
\newblock In {\em Proceedings of the 42nd ACM Symposium on Theory of Computing}, 2010.

\bibitem[JUW09]{JainUW09}
R.~Jain, S.~Upadhyay, and J.~Watrous.
\newblock Two-message quantum interactive proofs are in {PSPACE}.
\newblock In {\em Proceedings of the 50th Annual IEEE Symposium on Foundations
  of Computer Science}, pages 534--543, 2009.

\bibitem[JW09]{JainW09}
R.~Jain and J.~Watrous.
\newblock Parallel approximation of non-interactive zero-sum quantum games.
\newblock In {\em Proceedings of the 24th IEEE Conference on Computational
  Complexity}, pages 243--253, 2009.

\bibitem[Kal07]{Kale07}
S.~Kale.
\newblock {\em Efficient algorithms using the multiplicative weights update
  method}.
\newblock PhD thesis, Princeton University, 2007.

\bibitem[KKMV09]{KempeKMV09}
J.~Kempe, H.~Kobayashi, K.~Matsumoto, and T.~Vidick.
\newblock Using entanglement in quantum multi-prover interactive proofs.
\newblock {\em Computational Complexity}, 18(2):273--307, 2009.

\bibitem[KM03]{KobayashiM03}
H.~Kobayashi and K.~Matsumoto.
\newblock Quantum multi-prover interactive proof systems with limited prior
  entanglement.
\newblock {\em Journal of Computer and System Sciences}, 66(3):429--450, 2003.

\bibitem[Kob08]{Kobayashi08}
H.~Kobayashi.
\newblock General properties of quantum zero-knowledge proofs.
\newblock In {\em Proceedings of the Fifth IACR Theory of Cryptography
  Conference}, volume 4948 of {\em Lecture Notes in Computer Science}, pages
  107--124. Springer, 2008.

\bibitem[KW00]{KitaevW00}
A.~Kitaev and J.~Watrous.
\newblock Parallelization, amplification, and exponential time simulation of
  quantum interactive proof system.
\newblock In {\em Proceedings of the 32nd Annual ACM Symposium on Theory of
  Computing}, pages 608--617, 2000.

\bibitem[KSV02]{KitaevW02}
 A.~Kitaev, A.~Shen, M.~Vyalyi
\newblock Classical and Quantum Computation
\newblock  American Mathematical Society, 2002

\bibitem[LFKN92]{LundFKN92}
C.~Lund, L.~Fortnow, H.~Karloff, and N.~Nisan.
\newblock Algebraic methods for interactive proof systems.
\newblock {\em Journal of the ACM}, 39(4):859--868, 1992.

\bibitem[MW05]{MarriottW05}
C.~Marriott and J.~Watrous.
\newblock Quantum {Arthur-Merlin} games.
\newblock {\em Computational Complexity}, 14(2):122--152, 2005.

\bibitem[NC00]{NielsenC00}
M.~A. Nielsen and I.~L. Chuang.
\newblock {\em Quantum Computation and Quantum Information}.
\newblock Cambridge University Press, 2000.

%\bibitem[Nef94]{Neff94}
%C.~A. Neff.
%\newblock Specified precision polynomial root isolation is in {NC}.
%\newblock {\em Journal of Computer and System Sciences}, 48(3):429--463, 1994.

\bibitem[RW05]{RosgenW05}
B.~Rosgen and J.~Watrous.
\newblock On the hardness of distinguishing mixed-state quantum computations.
\newblock In {\em Proceedings of the 20th Conference on Computational
Complexity}, pages 344-354, 2005

\bibitem[Ros08]{Rosgen08}
B.~Rosgen.
\newblock  Distinguishing Short Quantum Computations.
\newblock In {\em Proceedings of the 25th STACS}, pages 597-608, 2008

\bibitem[Sha92]{Shamir92}
A.~Shamir.
\newblock {IP} $=$ {PSPACE}.
\newblock {\em Journal of the ACM}, 39(4):869--877, 1992.

\bibitem[She92]{Shen92}
A.~Shen.
\newblock {IP $=$ PSPACE}: simplified proof.
\newblock {\em Journal of the ACM}, 39(4):878--880, 1992.

\bibitem[vN28]{vN28}
J.~von Neumann.
\newblock Zur theorie der gesellschaftsspiele.
\newblock {\em Mathematische Annalen}, 100(1928):295--320, 1928

\bibitem[Wat99]{Watrous99-qip-focs}
J.~Watrous.
\newblock {PSPACE} has constant-round quantum interactive proof systems.
\newblock In {\em Proceedings of the 40th Annual IEEE Symposium on Foundations
  of Computer Science}, pages 112--119, 1999.

\bibitem[Wat08]{Watrous08}
J.~Watrous.
\newblock Lecture Notes for Theory of Quantum Information
\newblock Fall 2008.

\bibitem[Wat09a]{Watrous09a}
J.~Watrous.
\newblock Zero-knowledge against quantum attacks.
\newblock {\em SIAM Journal on Computing}, 39(1):25--58, 2009.

\bibitem[Wat09b]{Watrous09b}
J.~Watrous.
\newblock Semidefinite programs for completely bounded norms.
\newblock {\em Theory of Computing} 5: 11, 2009.

\bibitem[WK06]{WarmuthK06}
M.~Warmuth and D.~Kuzmin.
\newblock Online variance minimization.
\newblock In {\em Proceedings of the 19th Annual Conference on Learning
  Theory}, volume 4005 of {\em Lecture Notes in Computer Science}, pages
  514--528. Springer, 2006.

\bibitem[Wu10]{Wu10}
X.~Wu.
\newblock Parallized solutions to semidefinite programmings in
quantum complexity theory.
\newblock available at arXiv.org e-Print 1009.2211.

\end{thebibliography}

\end{document}